\documentclass[11pt,a4paper]{iopart}
\usepackage[T1]{fontenc}
\usepackage{lmodern}

\usepackage{iopams}  
\usepackage{graphicx}

\usepackage[breaklinks=true,colorlinks=true,linkcolor=magenta,urlcolor=blue,citecolor=blue]{hyperref}

\expandafter\let\csname equation*\endcsname\relax
\expandafter\let\csname endequation*\endcsname\relax
\usepackage{amsmath}
\usepackage{mathtools}
\mathtoolsset{showonlyrefs}

\usepackage{amsthm}
\newtheorem{theorem}{Theorem}
\newtheorem{theoremA}{Theorem}[section]
\newtheorem{corollary}{Corollary}

\usepackage{xcolor}
\definecolor{editcol}{HTML}{E34A33}

\newcommand{\edit}[1]{#1}

\newcommand{\prob}{\mathsf{P}}

\newcommand{\Ell}{\mathbb{L}}
\newcommand{\JJ}{\mathbb{J}}
\newcommand{\Aop}{\mathbb{A}}

\newcommand{\dd}{\text{d}}

\newcommand{\deriv}[2]{\frac{\text{d} #1}{\text{d} #2}}

\newcommand{\eps}{\varepsilon}

\renewcommand{\P}{\mathbb{P}}

\newcommand{\TT}{\mathbb{T}}

\newcommand{\subfig}[1]{\textbf{\sffamily(#1)}}

\usepackage[square,sort&compress,numbers]{natbib}
\begin{document}

\title[Jump Locations of Jump-Diffusion]{Jump Locations of Jump-Diffusion Processes with State-Dependent Rates}




\author[cor1]{Christopher E. Miles$^{1}$}
\ead{\mailto{miles@math.utah.edu}}

\author{James P. Keener$^{1,2}$}
\qquad

\address{$^1$Department of Mathematics, University of Utah}
\address{$^2$Department of Bioengineering, University of Utah}

\begin{abstract}
We propose a general framework for studying statistics of jump-diffusion systems driven by both Brownian noise (diffusion) and a jump process with state-dependent intensity. Of particular natural interest in many physical systems are the jump locations: the system evaluated at the jump times. As an example, this could be the voltage at which a neuron fires, or the so-called ``threshold voltage.''  However, the state-dependence of the jump rate provides direct coupling between the diffusion and jump components, making it difficult to disentangle the two to study individually. In this work, we provide an iterative map formulation of the sequence of distributions of jump locations. The distributions computed by this map can be used to elucidate other interesting quantities about the process, including statistics of the interjump times. Ultimately, the limit of the map reveals that knowledge of the stationary distribution of the full process is sufficient to recover (but not necessarily equal to) the distribution of jump locations. We propose two biophysical examples to illustrate the use of this framework to provide insight about a system. We find that a sharp threshold voltage emerges robustly in a simple stochastic integrate-and-fire neuronal model. The interplay between the two sources of noise is also investigated in a stepping model of molecular motor in intracellular transport pulling a diffusive cargo.
\end{abstract}

\maketitle

\section{Introduction}

 Stochastic models driven by both Brownian noise (diffusion) and a discrete jump component, or so-called \textit{jump-diffusion} models, have been used to describe a wide variety of phenomena. Perhaps most prominently, jump-diffusion has seen widespread use in mathematical finance \cite{Yan2006a,Kou2002,Kou2007,Figueroa-Lopez2012,Bjork1997} to describe frequent small transactions with occasional larger movements. Jump-diffusion has also served useful in mathematical biology, describing the integrate-and-fire nature of neuronal dynamics \cite{Gerstner2002,Sacerdote2013,DiMaio2004}. Other applications include biophysical descriptions of movement of chromosomes \cite{Shtylla2010}, interaction between soil moisture and rainfall events \cite{Daly2006b}, and the occurrence of radio pulsar glitches \cite{Fulgenzi2017}. Also relevant to the results described in this paper is a subset of jump-diffusion models that neglect diffusive noise and are driven by deterministic dynamics between jumps. This type of model appears in an equally eclectic variety of applications \cite{Wu2008,Daly2007a,DOdorico2006a,Bartlett2015a}.

Of particular interest in this work are jump-diffusion processes with state-dependent jump rates. This subset is seen in all the aforementioned applications, including finance \cite{Bjork1997,Glasserman2004}, ecology \cite{DOdorico2006a}, biology \cite{Daly2006c}, and astronomy \cite{Fulgenzi2017}. This is often a natural supposition to make of a model, as the rate at which the jump event occurs may not remain constant on the timescale of interest, but instead depend on the state of the system.  For instance, a neuron firing is dependent on the current voltage or a financial asset is often more likely to crash as it rises in price. 

For many of these applications, the inclusion of a jump component means that it is a significant feature to the modeler. Thus, it is natural to ask: when (in some sense) do these jumps occur? While this often leads to studying the interjump \textit{times}, a different possible interpretation (and the focus of this work) is to study the \textit{jump locations}: the system evaluated at the jump times. This could be the voltage at which the neuron fires (the ``threshold voltage'') or the price at which a stock crashes. Because the jump intensity is state-dependent, and the state itself is governed by a stochastic process, these locations are inherently random and difficult to disentangle from the full process. This raises another question: in what sense is knowledge (e.g. statistics) of the full process equivalent to knowledge of the jump locations? Does knowledge of the jump locations provide other insight to the behavior of the system? In this work, we seek to investigate these issues.

Previous studies have thoroughly investigated jump-diffusion models  for constant jump rates \cite{Daly2006b,Hanson2007}, as the two sources of noise can be independently superimposed. The interjump time distributions for processes with no diffusive noise have also been studied \cite{Daly2006c,Daly2007a} primarily using renewal theory or exploiting the deterministic nature between jumps. However, in general,  state-dependent jump-diffusion is  \textit{not} deterministic between jumps, nor a renewal process, as sequential distributions are not independent, hence previous approaches do not apply. The theory of Cox processes \cite{Cox1955} is also well-established and refers to  a diffusion process driving a state-dependent Poisson process, but this does not complete the feedback loop of the jump process further modifying the state, which our framework allows. Diffusion with switching, which bears a resemblance to jump-diffusion, has also been studied \cite{Lawley2015} but is distinctly different in the role of discrete noise and not addressed in this work. Although efficient Monte Carlo methods exist for jump-diffusion and L\'{e}vy processes \cite{Xia2012,Glasserman2004}, the jump component may be rare and therefore expensive to find an accurate distribution empirically. The solution to the full PDE describing this process can also be computed, but it is seemingly unknown how to disentangle the jump component alone from this description.

In this work, we present a general (accommodating a wide variety of models) framework for studying jump-diffusion systems and  focus particularly on studying the sequence of jump locations. We present an iterative map formulation explicitly describing the distribution of the $i$th jump location.  Fron this sequence of jump locations, statistics of the interjump times can be extracted even with diffusion, overcoming a limitation of previous works \cite{Daly2007a}. By taking the limit of the map, assuming the process reaches stationery, we find that the density of the full process \textit{differs} from that of the jump locations if and only if the jump rate is state-dependent. An explicit relationship between these two distributions is established, meaning that knowledge of one immediately provides knowledge of the other. 

A few examples are discussed, illustrating when this framework can be used to elucidate interesting behavior of a system. The first example, a stochastic model of neuronal integrate-and-fire, is used to demonstrate that the jump locations themselves (in this case, the firing voltages) may be interesting. For the second example, we propose a simple model of a molecular motor (such as kinesin \cite{Howard2001}) taking force-dependent steps pulling a diffusive cargo. Because the stepping rate is state dependent, but steps (and cargo diffusion) also modify the state, the effective stepping rate (a proxy for the motor's ability to produce transport) is non-trivial and can be studied using the proposed framework.

\section{Formulation}

\subsection{Jump-diffusion}
Let $X_t$ denote a state-dependent jump-diffusion process and $\tilde{p}(x,t)$ be the probability density (PDF) of the full process, \edit{defined formally as}
\begin{equation}
	\edit{\tilde{p}(x,t) \, \dd x \coloneqq \prob\left[X_t \in (x, x+ \dd x)\right]}. \label{eq:tildep}
\end{equation}
 The evolution of the density is described by the forward Champman-Kolmorgov equation \cite{Gardiner2009,Bressloff2014}
\begin{equation}
	\partial_t \tilde{p}(x,t) = \Ell \tilde{p} - \lambda(x)\tilde{p} + \JJ  \lambda \tilde{p}, \label{eq:reinject}
\end{equation}

where the operator $\Ell$ characterizes the diffusion component, $\lambda(x)$ is the state-dependent intensity at which the jump process occurs, and the operator $\JJ$ describes the jump. A feature we seek to emphasize is the state dependence of the jump-rate $\lambda(x)$, resulting in direct coupling between the diffusion and jump components of the process. 

\begin{figure}[h]\centering
\includegraphics[width=.66\textwidth]{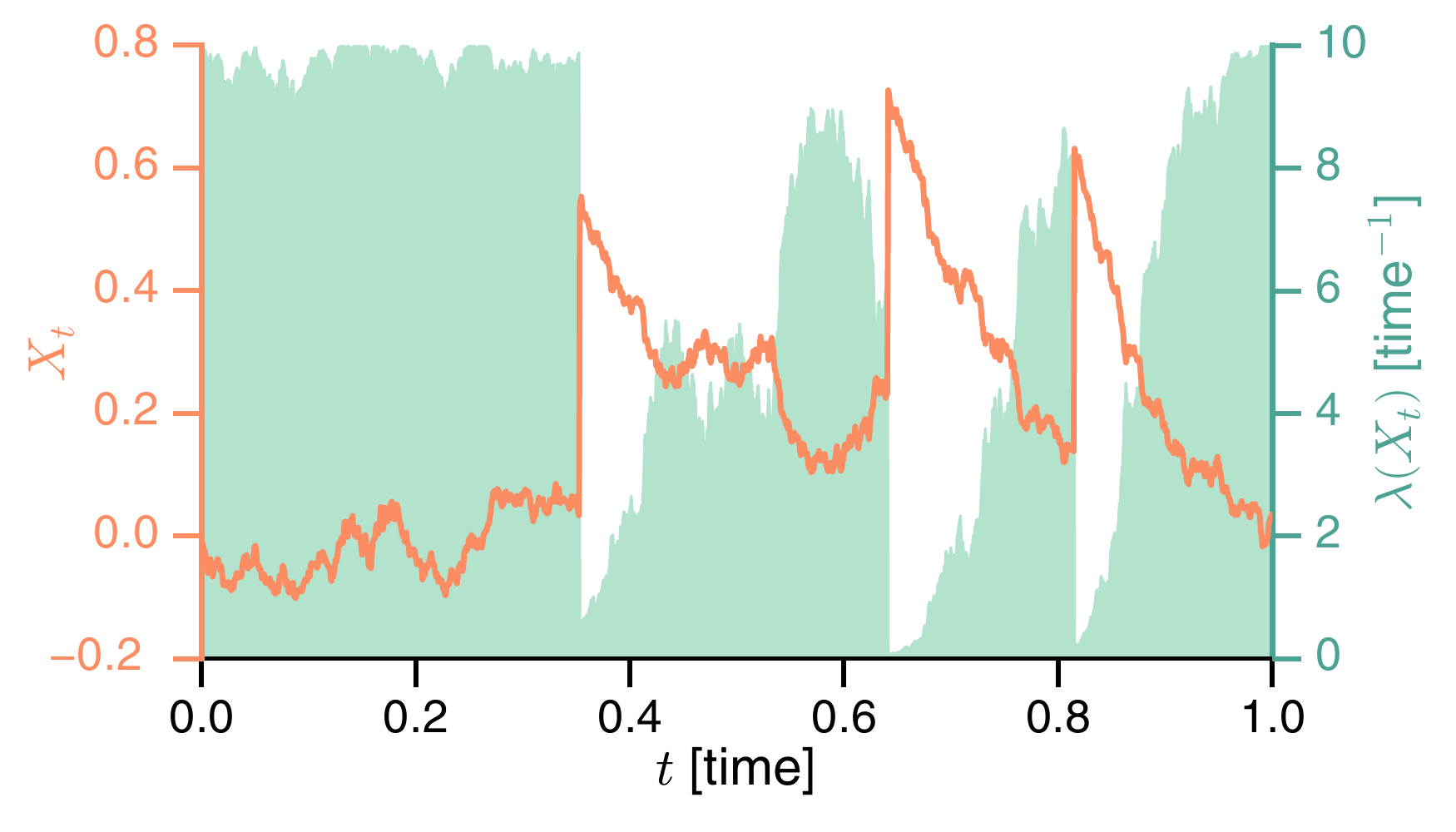}\centering
  \protect\centering\caption{An example realization  of a jump-diffusion process $X_t$ with $\Ell p = \partial_x\left\{axp \right\} + Dp_{xx}$, $\lambda(x) = \alpha \exp\left\{-x^2/\beta\right\}$ and $\JJ p = p(x-\Delta)$. Also shown is the state dependent jump rate $\lambda(x)$. In this particular realization, 3 jumps occur. Parameter values used $a=.1, D=1, \alpha=10, \beta=10, \Delta=1$.}
                \label{fig:ex}
\end{figure}
The particular choices of $\Ell, \JJ, \lambda$ significantly change the behavior and characteristics of the process, but the description \eqref{eq:reinject} is flexible enough to accommodate a wide variety of models, upon which we now elaborate.

\subsection{Diffusion component}
Consider the classical one-dimensional It\^{o} SDE driven by Brownian motion described by
\begin{equation}
	\dd Y_t = A(Y_t)\, \dd t + \edit{\sqrt{2D(Y_t)}} \, \dd W_t, \label{eq:langevin}
\end{equation}
where $W_t$ is a Brownian motion. The Fokker-Planck (forward Chapman-Kolmogorov) equation describing the evolution of the probability density function $p(y,t)$ corresponding to the path-wise description \eqref{eq:langevin} is 
\begin{equation}
	\partial_t p(y,t) = \Ell p \coloneqq -\partial_y\left\{A(y) p\right\} + \edit{\partial_{yy} \left \{ D(y) p\right\}} .
\end{equation}
We use the \textit{Fokker-Planck operator} $\Ell$ to characterize the diffusion process for the remainder of the paper. Although the It\^{o} interpretation is used here, we could also consider a Stratonovich interpretation by modifying the particular details of $\Ell$, as this description still produces a differential operator in the forward equation. 

\subsection{Jump component}
Let the jump process be an inhomogeneous (state-dependent) Poisson process with intensity (rate) $\lambda(X_t)$. Also necessary is a description of what occurs at the jumps, sometimes referred to as the \textit{reset map} \cite{DeVille2014}.  The behavior of the reset map is characterized by the \textit{jump operator} $\JJ$, which is a probability density flux, ensuring that \eqref{eq:reinject} indeed describes the evolution of a probability density. We now describe a few possible choices of this operator, letting $\tau$ denote the jump time.
\subsubsection{Jump operator examples}
\begin{enumerate}
\item \textit{Constant jump size}. At the jump times, the process increments by a fixed amount $\Delta$, so
\begin{equation}
	X_{\tau_{+}} = X_{\tau_-} + \Delta.
\end{equation}
Then, the corresponding jump operator $\JJ$ is a shift by that fixed quantity
\begin{equation}
	\JJ p\edit{(x,t)} \coloneqq p(x-\Delta,t).	
\end{equation}
\item \textit{Reset}. Fundamentally different than the previous example, rather than jumping by a fixed displacement, the process resets to a specific position, $\eta$, so that
\begin{equation}
	X_{\tau_+} = \eta.
\end{equation}
The corresponding $\JJ$  operator is 
\begin{equation}
	\JJ p\edit{(x,t)} \coloneqq \delta(x-\eta) \int_{-\infty}^{\infty} p(x,t) \, \dd x,
\end{equation}
where $\delta(x)$ is the Dirac delta function. The integral scaling is necessary to preserve \edit{probability fluxes: $\int_{-\infty}^{\infty} \JJ p \, \dd x = \int_{-\infty}^{\infty} p \, \dd x$.  }
\item \textit{Random jump size}. A generalization of the first example, the particle can jump a random displacement $\Delta$, where the size of the jump is described by the probability density $\Delta \sim \mu(\Delta,x)$, which can also be state-dependent. The $\JJ$ operator is
\begin{equation}
	\JJ p\edit{(x,t)} \coloneqq \int_{-\infty}^{\infty} p(x-\Delta,t) \mu(\Delta,x) \, \dd \Delta.
\end{equation}
\item \textit{Other examples}. Although not directly considered in this paper, maps of the form 
\begin{equation}
	X_{\tau_+} = \gamma X_{\tau_-} \qquad \implies \qquad \mathbb{J}p\edit{(x,t)} \coloneqq \gamma p(\gamma x,t)
\end{equation}
have been studied elsewhere \cite{DeVille2014} and can be formulated in our framework in the described manner. A generalization of this case could also be made to allow for jumps to a random location, but this does not seem to provide additional interesting structure.
\end{enumerate}

\edit{
While diffusion with resetting fits within the proposed framework, this topic has a long, rich history of study \cite{Evans2011,Montero2017}. Included  within this body of literature are investigations of: time-dependent resetting ($t$-dependent $\lambda$) \cite{Pal2016}, optimal resetting \cite{Evans2011b}, resetting in bounded domains \cite{Christou2015b}, resets with drift \cite{Montero2013}, characterizations of the stationary behavior \cite{Mendez2016}, and path-integral formulations \cite{Roldan2017}. We do not seek to propose our framework as an alternative to this rich body of literature, but merely identify that stochastic resetting fits within a family of models that otherwise lack this same level of attention and can be studied with the tools presented here. }

\section{Results}

In this section, we present the general theoretical results of the paper regarding the distributions of jump locations and interjump times.

\subsection{Jump distributions}
\edit{We begin by providing intuition for the construction of the sequence of jump locations. Rather than studying the full process including diffusion and jumps, a more convenient process} to study is a survival formulation between jumps, which we denote the \textit{absorbing process},
\begin{equation} \label{eq:absorb}
	\begin{cases}
	\partial_t p(x,t) = \Ell p - \lambda p\\
	\partial_t q(x,t) = \lambda p.
	\end{cases}
\end{equation}
Note that we distinguish between this and the full process $\tilde{p}(x,t)$, which we refer to henceforth as the \textit{reinjected process}. While the absorbing process does not capture all the behavior of the full re-injected process (as it does not contain any information about $\JJ$), the absorbing process fixes the distribution (in both space and time) of jump events and consequently serves more fruitful in disentangling the jump component from the full process.

\begin{theorem}\label{thm:1}
The densities of the jump location $p_{\ell}(x)$ and jump time $p_{\tau}(t)$ of \eqref{eq:absorb} are described by
\begin{equation}
	p_{\ell}(x) = q(x,\infty) = \int_0^{\infty} \lambda(x) p(x,t) \, \dd t \end{equation}
	and 
	\begin{equation}
	 p_{\tau}(t)  = \int_{-\infty}^{\infty} \lambda(x) p(x,t) \, \dd x.
\end{equation}
\end{theorem}
\begin{proof}
Define the survival probability $s(t)$ to be the probability that the jump \textit{has not} occurred by time $\tau$
\begin{equation}
	s(t) \coloneqq \mathbb{P}[\tau > t] = \int_{-\infty}^{\infty}p(x,t) \, \dd x,
\end{equation}
which means the first jump time density $p_\tau(t)$ is then 
\begin{equation}
	p_\tau(t) = -\deriv{s}{t} = -\int_{-\infty}^{\infty} \partial_t p \, \dd x.
\end{equation}
Then, using \eqref{eq:absorb} and the fact that $p \to 0 $ as $x \to \pm \infty$, we obtain
\begin{equation}
	p_\tau(t) = \int_{\infty}^{\infty} \lambda p \, \dd x = \int_{\infty}^{\infty} \partial_t q \, \dd x.
\end{equation}
Due to the lack of spatial flux in $q(x,t)$, the density of  exit locations is  $q(x,\infty)$, which \edit{can be} obtained by integrating over all possible jump times
\begin{equation}
	p_\ell(x) = \int_0^\infty \partial_t q \, \dd t = q(x,\infty) - q(x,0) = q(x,\infty)
\end{equation}
by noting that $q(x,0) \equiv 0$. \edit{For this proof (and the remainder of the results), we assume that $\lambda, \Ell, \JJ$ are chosen with sufficient regularity such that there are no explosions (i.e. the intensity stays finite. For more details on these conditions, see \cite{Bally2017}.}
\end{proof}
Although this provides useful information about the distributions (in time and space) at which a particular jump occurs, we are interested in studying the distributions of all jumps. This requires including the information embedded in $\mathbb{J}$ about how to re-inject the particle. triggered. \edit{A similar idea of studying the system between jumps was proposed in \cite{Montero2016}, in which the authors study the record statistics of the so-called Sisyphus random walk. However, the Sisyphus random walk is a discrete random walk on a lattice, rather than a continuous diffusion as discussed here, meaning the results are not directly applicable but strongly paralleled.}
\subsection{Jump location sequential mapping}

We ystudy the interjump dynamics by noting that between jumps, the full process \eqref{eq:reinject} can be described by the absorbing process \eqref{eq:absorb}. Define $t_i$ to be the $i$th jump time and let $X_i \coloneqq X_{t_i}$ be the $i$th jump location. \edit{Let $p_i(x)$ be the probability density of $X_{t_i}$ formally defined to be, 
\begin{equation} 
	p_i(x) \dd x \coloneqq \prob[X_{t}\in(x, x+\dd x) \, |\, t=t_i].
\end{equation}
\begin{theorem}
The distribution of the $i$th jump location, $p_i(x)$  is described by the following iterative relation
\end{theorem}\vspace{-.5in}
\begin{equation} \label{eq:pi}
\begin{cases}
\partial_t \hat{p}_i(x,t) = \Ell \hat{p}_i - \lambda \hat{p}_i \\
\hat{p}_i(x,0) = \begin{cases} \JJ p_{i-1}  & i > 1 \\ p_0 & i = 1\end{cases}\\
p_{i+1}(x) = \int_{0}^{\infty} \lambda \hat{p}_i \, \dd t,
\end{cases}
\end{equation}
where $p_0$ be some known starting distribution of the process. }

This construction follows from directly from \textbf{Theorem \ref{thm:1}}. Effectively, $\hat{p}_i$ serves as intermediate quantity tracking the distribution of all possible jump locations (and times) for that particular iterate. \edit{That is, $\hat{p}_i$ is the survival density (in both time and space), tracking the process between jumps, so 
\begin{equation}
 	\hat{p}_i(x,t) \dd x \coloneqq \prob[X_{t+t_{i-1}} \in (x, x+\dd x)\, , \, t+t_{i-1} < t_{i}].
 \end{equation}} Then, to start the next iterate, the distribution of jump locations must be modified by the jump procedure characterized by $\JJ$. Although tracking the jump locations rather than the jump times seems counterintuitive at first, no natural analogous formulation is apparent. 

Define the more convenient quantity to study
\begin{equation}u_i \coloneqq p_i(x)/\lambda(x). \label{eq:u_idef}\end{equation}
 If $\lambda(x) =0$ for some  $x$, then necessarily $p_i(x) = 0$ since $\lambda(x) =0 $ implies no jump can occur at this location, hence this quotient causes no difficulties. 
\begin{theorem}
The description \eqref{eq:pi} is equivalent to
\begin{equation}
	\TT u_{i+1} \coloneqq \left [\lambda(x) - \Ell \right] u_{i+1} = \JJ \lambda u_i,
\end{equation}
or more explicitly 
\begin{equation}
	u_{i+1} = \TT^{-1} \JJ \lambda u_i. \label{eq:qi}
\end{equation}
\end{theorem}
\begin{proof}
Integrating both sides of \eqref{eq:pi} with respect to $t$ and noting that $\hat{p}_i(x,\infty)=0$ we are left with
\begin{equation}
	\hat{p}_i(x,0) = \int_{0}^{\infty}\Ell \hat{p}_i \, \dd t - \lambda(x) \int_0^{\infty} \edit{\hat{p}_i \, \dd t}.
\end{equation}
The linear operator $\Ell$ is a differential operator in $x$ and consequently commutes with the time integral. Using the initial condition and $p_i/\lambda = \int_0^{\infty} \hat{p}_i \, \dd t$, we obtain the desired result.
\end{proof}
The map \eqref{eq:qi} provides an explicit description for the sequence
\begin{equation}
	u_1(x) \to u_2(x) \to \cdots \to u_\star(x) \to u_\star(x) \to \cdots, \label{eq:iterates}
\end{equation}
where, we assume $u_\star$ is the fixed point of of the  map \eqref{eq:qi} and exists. Although we seek  to not dwell on this aspect, a natural question to ask is under what conditions does $u_\star$ exist? Equivalently, when does the the process reach stationarity? In \textbf{\ref{sect:spectral}}, we provide a brief commentary about how the relationship \eqref{eq:qi} can be thought of as an iterated linear non-negative integral operator which can be studied as such. Results from theoretical ecology literature are then cited which provide a heuristic analysis of conditions for convergence. 

We proceed assuming that \eqref{eq:qi} has a fixed point, which must be of the following form. 
\begin{theorem}
The stationary distribution of the jump locations $p_\star(x)$ is described by 
\begin{equation}
	0 = \Ell u_\star - \lambda u_\star + \JJ \lambda u_\star, \label{eq:qstar}
\end{equation}
where $u_\star \coloneqq p_\star/\lambda$.
\end{theorem}
\begin{corollary} \label{cor:1}
The stationary distribution of the jump locations $p_\star$ is the same as the stationary distribution of the full process if and only if $\lambda$ is constant. 
\end{corollary}
\begin{proof}
This is an immediate consequence of taking the full process to be in stationarity, so $\dd \hat{p} / \dd t =0$ in \eqref{eq:reinject}, which satisfies
\begin{equation}
0 = \Ell \hat{p}_{s} - \lambda \hat{p}_{s} + \JJ \lambda \hat{p}_{s}.  \label{eq:pss}
\end{equation}
This is exactly the same relationship as \eqref{eq:qstar}. Recalling that $ \lambda u_{\star} = p_{\star}$ and that both $p_{\star}$ and $\hat{p}_{s}$ are  probability densities, the only way that $p_{\star} = \hat{p}_{s}$ is if $u_\star$ is a rescaling of $p_\star$ or, equivalently, $\lambda$ is constant. 
\end{proof}
Elaborating a bit on this relationship: while \eqref{eq:pss} and \eqref{eq:qstar} may appear to produce the same result, the solutions to each require different scalings. Since $\hat{p}_{s}$ is a probability density, it must be that 
\begin{equation}
\int_{-\infty}^{\infty} \hat{p}_{s} \, \dd x =1.  \label{eq:intPSS}
\end{equation}
 However, $u_\star \lambda = p^\star$, where $p_\star$ is a probability density, so this means that
 \begin{equation}
\int_{-\infty}^{\infty} u_\star(x) \lambda(x) \, \dd x  = 1. \label{eq:intUL}
 \end{equation} 

It is worth pointing out that \cite{Daly2007a} derives similar results relating the jump locations and jump times, however, the justification used there  depends  on the trajectories being deterministic between jumps, i.e. without diffusion.

\subsection{Moments}

Define  $\tau_i$ to be the $i$th interjump time, $\tau_i \coloneqq t_{i}-t_{i-1}$. Also, let $\tau_\star$ be the interjump time for the stationary process.
\begin{theorem}
The mean interjump time $\tau_i$ can be recovered from the distribution of the $i$th jump location and is
\begin{equation}
	\langle \tau_i \rangle = \int_{-\infty}^{\infty} u_i \, \dd x =  \int_{-\infty}^\infty \frac{p_i(x)}{\lambda(x)} \, \dd x. \label{eq:tauimean}
 \end{equation}
 Also, the mean stationary interjump time can be computed from the stationary distribution of the jump locations $p_\star$
 \begin{equation}
 	\langle \tau_\star \rangle  =  \int_{-\infty}^{\infty} u_\star \, \dd x = \int_{-\infty}^\infty \frac{p_\star(x)}{\lambda(x)} \, \dd x. \label{eq:taustar}
 \end{equation}
\end{theorem}

\begin{proof}
We integrate both sides of \eqref{eq:pi} with respect to $x$, again noting that $\hat{p}_i \to 0$ as $x \to \pm \infty$, resulting in
\begin{equation}
\partial_t  \int_{-\infty}^{\infty} \hat{p}_i \, \dd x = -\int_{-\infty}^{\infty} \lambda \hat{p}_i \, \dd x. 
\end{equation}
However, from \textbf{Theorem \ref{thm:1}}, we see that the right-hand side is exactly the distribution of the interjump time $p_{\tau_i}$,
so we have
\begin{equation}
\partial_t  \int_{-\infty}^{\infty} \hat{p}_i \, \dd x  = -p_{\tau_i}(t) \label{eq:ptaui}
\end{equation} 
Taking the mean on both sides with respect to $\tau_i$,
\begin{equation}
\langle \tau_i \rangle =  -\int_{0}^{\infty} \int_{-\infty}^{\infty}t \partial_t \hat{p}_i \, \dd x \, \dd t,
\end{equation}
after integrating by parts and noting that $\int_0^{\infty} \hat{p}_i \, \dd t = p_i/\lambda$, we get the desired result.
\end{proof}
While the first moment (mean) of the interjump time can be computed directly with a quadrature of the presumed known jump location $p_i$, the higher order moments are less straightforward. \textbf{\ref{sect:higher_moments}} describes how, in theory, knowledge of the $p_i$ can be used to extract higher order moments of $\tau_i$. Solving for higher order moments using this approach requires solving a hierarchy of differential equations, which in practice, may not be so feasible. 
However, a more practical numerical approach may be to solve for $p_i$ and then run the absorbing \eqref{eq:absorb} process to extract explicitly $p_{\tau_i}$ using the relationships from \textbf{Theorem  \ref{thm:1}}.

\section{Examples}

\subsection{Neuronal integrate-and-fire}

The behavior of individual neurons can roughly be thought of as an ``integration'' process, which builds up voltage, and then ``fires,'' which describes the release of this voltage into an action potential \cite{Gerstner2002}, the full process of which is known as \textit{integrate-and-fire}. The most classical version of this model involves a deterministic buildup of voltage until a fixed firing voltage threshold is reached. 

Incorporating noise in various ways into this class of models has a rich body of literature \cite{Sacerdote2013,Plesser2000}. Some approaches include a so-called ``leaky'' neuron, one that has diffusive noise in the integration phase \cite{Dumont2016,DiMaio2004,Fourcaud2002}, whereas others regard the threshold itself as stochastic \cite{Braun2015,Daly2006c,Lindner2005b,Dumont2016}. Our work provides a natural framework to consider both sources of noise and their effect.

We propose a simple model of neuronal integrate-and-fire. While this model is greatly simplified from the actual physiology, we take this approach to show that a minimal model, stripped of considerable details, is still able to produce interesting emergent behavior. In previous models with a stochastic threshold, there is some inherent threshold, say, $v_0$, where the firing rate $\lambda(v)$ is taken to be a Gaussian centered around $v_0$ \cite{Daly2006c,Braun2015}. The justification for this that models such as Hodgkin-Huxley \cite{Keener2008}, which account for more fine-grained detail, predict a distinct firing threshold $v_0$. In our model, we find that this sharp threshold $v_0$ can also arise from stochasticity alone, even when no inherent threshold is defined in the model. 

Let $V_t$ denote the voltage at time $t$. Then, in the context of this framework, we take the forms
\begin{equation}
\Ell p = -\partial_v \{ \alpha p\}+ \edit{D \partial_{vv} p}, \qquad \lambda(v) = \gamma e^{\beta v}, \qquad \JJ p = \delta(v) \int_{-\infty}^{\infty}p(v,t) \, \dd v.
\end{equation}
In words:  we take the buildup to be at a constant rate $\alpha$ and leaky with constant diffusivity. While it is typically more common to take evolution of the form $\dot{v} = -a + bv$, this defines an inherent voltage $\tilde{v}=b/a$, which we deliberately omit. The result of making this change is also qualitatively negligible.  For the jump component, we take $\mathbb{J}$ to be the reset to zero case, a full fire. The rate at which the firing occurs is assumed to be monotonically increasing, where again, we emphasize the fundamental difference between this form and previously considered, where $\lambda$ is a Gaussian centered around some voltage $v_0$. In our model, voltage simply builds up, and as it builds up, the system is more likely to fire. There is no prescribed firing threshold as in classical work and its extensions. The justification for this form can be thought of as the underlying mechanism for firing: a discrete birth-death process (also with voltage dependent rates) of ion channels. At some fixed voltage $v$, the birth-death process is more likely to flip entirely ``on'' and cause the neuron to fire. The relationship between models with explicit ion channels and integrate-and-fire has been explored previously \cite{Ditlevsen2013}.  

\begin{figure}[h]
\includegraphics[width=.45\textwidth]{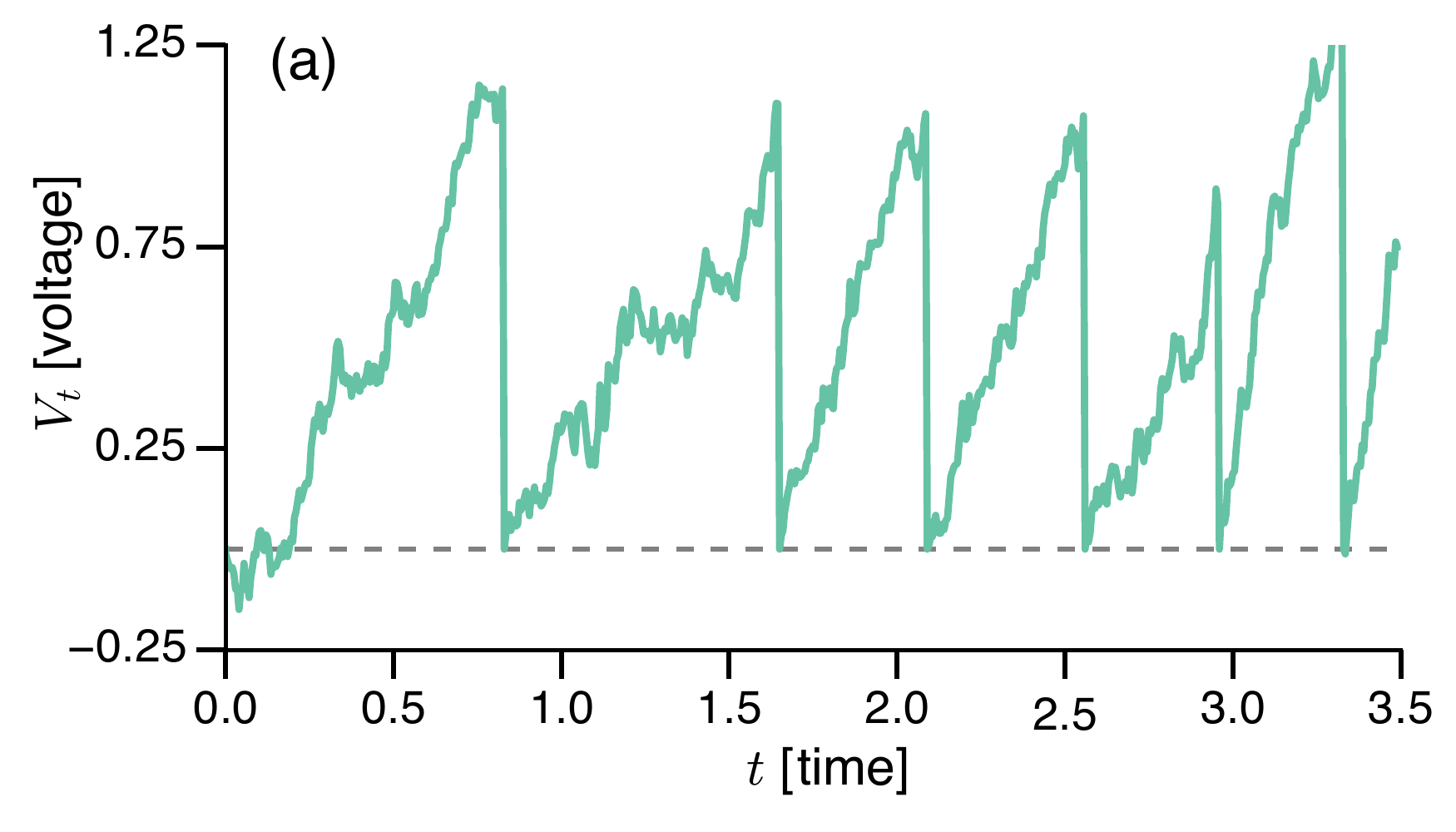}
\includegraphics[width=.6\textwidth]{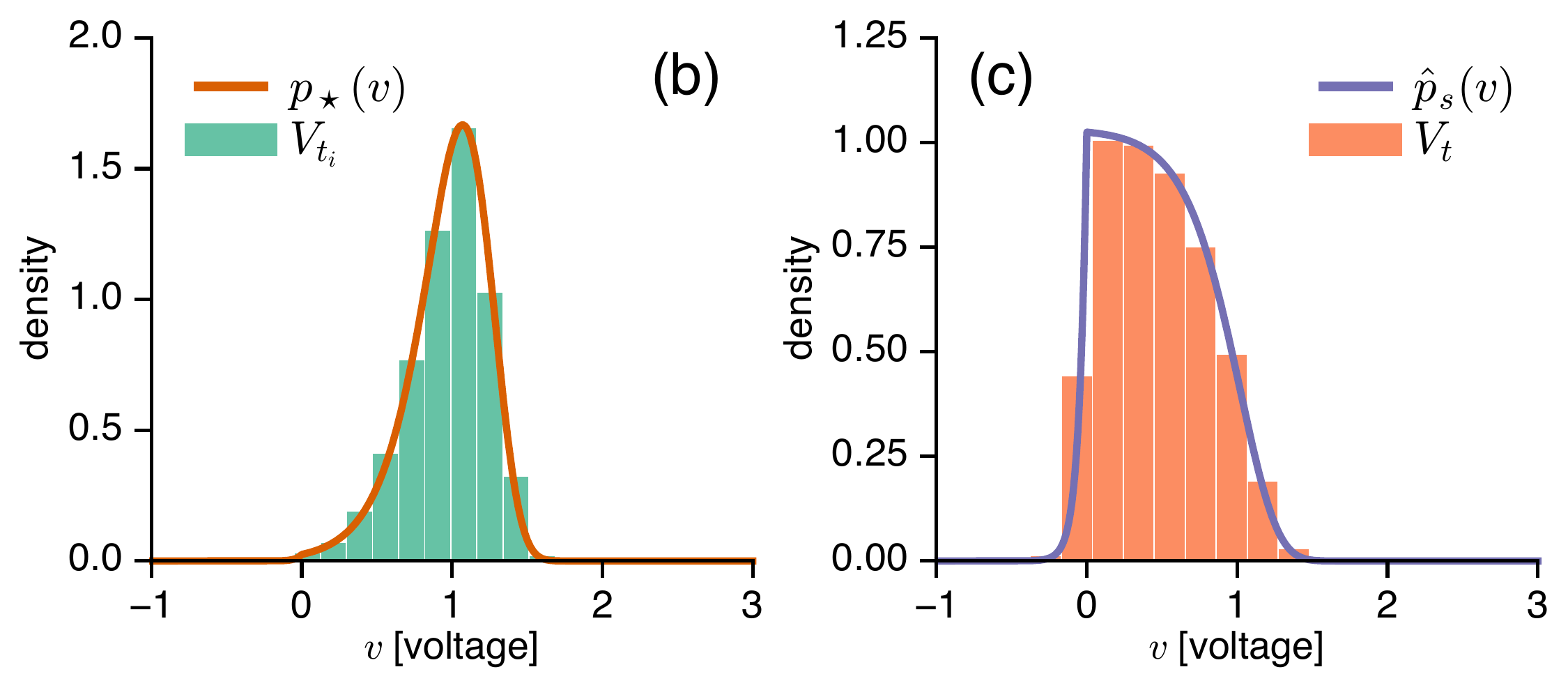}
\protect\caption{\subfig{a} an example trajectory of the proposed integrate-and-fire model. When the neuron fires, it resets back to $V=0$ (dashed, gray). \textbf{\sffamily(b)} the resulting stochastic firing voltages from Monte Carlo simulation (bars) and theoretically predicted (line). \textbf{\sffamily(c)} the stationary density of the full process from Monte Carlo simulation (bars) and predicted (line).   \label{fig:neuronfig}}
\end{figure}

An example simulation of the model can be seen in \textbf{Figure \ref{fig:neuronfig}}(a) with parameter values $\alpha =2, D=.1, \gamma =.5, \beta =2$. In the simulation, voltage builds up in a noisy (leaky) manner, and then the jump process fires, resetting back to $v=0$, typical of an integrate-and-fire model. However, in this case, we emphasize that the firing voltage is inherently random and therefore of interest to study.  Then, in stationarity, the fundamental  relation in this work \eqref{eq:qstar} becomes (noting $x \to v$)
\begin{equation}
	0 = -\partial_v \left\{\alpha u_\star\right\} + D\partial_{vv} u_\star - \gamma e^{\beta x} u_\star + \delta(v) \int_{-\infty}^{\infty}\gamma e^{\beta v} u_\star \, \dd v. \label{eq:u_ex1}
\end{equation}
From \eqref{eq:intUL}, the integral scaling on the $\delta$ function is equal to $1$, 
\begin{equation}
	\delta(v) = -\partial_v \left\{\alpha u_\star\right\} + D\partial_{vv} u_\star - \gamma e^{\beta v} u_\star. \label{eq:ustar_ex1}
\end{equation}
This can be computed in a similar manner to a Green's function, noting that for $v \neq 0$, we have
\begin{equation}
	0 = -\partial_v \left\{\alpha u_\star\right\} + D\partial_{vv} u_\star - \gamma e^{\beta v} u_\star, \label{eq:unot0}
\end{equation}
and integrating \eqref{eq:ustar_ex1} from $(-\eps,\eps)$ and using the fact that $u_\star$ must be continuous, we get the matching condition
\begin{equation}
	D\left\{ u'(0^+)  - u'(0^-)\right\} = -1. \label{eq:jumpcondit}
\end{equation}
Solving \eqref{eq:unot0}, we get that our solution is of the form 
\begin{equation} \label{eq:uform}
	u_\star = \begin{cases} c_1 u_L(v) \coloneqq c_1 e^{\frac{\alpha v}{2D}} I_{\alpha/(\beta D)} \left(\frac{2\sqrt{D\gamma e^{\beta v}}}{\beta D}\right) & v < 0 \
	\\
	c_2u_R(v) \coloneqq c_2 e^{\frac{\alpha v}{2D}} K_{- \alpha /(\beta D)} \left(\frac{2\sqrt{D\gamma e^{\beta v}}}{\beta D}\right)  & v >0, 	
	\end{cases}
\end{equation}
where $I, K$ are modified Bessel functions \cite{Abramowitz1966} such that $I \to 0$ as $v \to -\infty$ and $K \to 0$ as $v \to \infty$ (and are linearly independent). The jump condition \eqref{eq:jumpcondit} then becomes
\begin{equation}
	D\left\{c_2u_R'(0) - c_1u_L'(0)\right\} = -1. \label{eq:condit1}
\end{equation}
However, we also need to impose continuity of $u_\star$, so we also have the requirement
\begin{equation}
	c_1u_L(0) = c_2u_R(0). \label{eq:condit2}
\end{equation}
The conditions \eqref{eq:condit1},\eqref{eq:condit2} provide us two equations for two unknowns, which yield
\begin{equation}
	c_1 = \frac{2}{D\beta } 2 K_{\alpha/(D\beta)}\left(\frac{2\gamma}{\beta\sqrt{D\gamma}}\right), \qquad c_2 = \frac{2}{\beta D} 2 I_{\alpha/(\beta D)}\left(\frac{2\gamma}{\beta\sqrt{D\gamma}}\right).
\end{equation}
Thus far, our solution is defined only up to a constant, but we know that $u_\star$ must satisfy the scaling \eqref{eq:intUL}, which our choice of $c_1,c_2$ serendipitously already satisfy. By computing $u_\star$ explicitly, we can use \eqref{eq:u_idef} to immediately obtain the distribution of jump locations $p_\star$. Finally, from  \textbf{Corollary \ref{cor:1}}, we also have the stationary density of the full process, $\hat{p}_{s}$, which is a rescaling of $q_\star$ such that \eqref{eq:intPSS} is satisfied. 

The stationary density $\hat{p}_s$ of the full process and the stationary jump distribution $p_\star$ are shown in \textbf{Figure \ref{fig:neuronfig}}(b,c). From this, we can see that $\hat{p}_s$, the stationary density of the full process provides no interesting information. It is relatively uniform through some range of voltages. If a PDE approach were taken to provide information about the full process, this is all that would be available. However, the jump locations $p_\star$, are noteworthy. Despite prescribing no explicit voltage threshold for firing, the state-dependent nature of the firing rate yields an effective threshold ($v_0 \approx 1$ for these parameter values) and persists through a wide range range of parameters. Hence, this is another, novel justification for models that take a deterministic firing threshold, as our minimal integrate-and-fire model produces this feature as a product of stochasticity, which is made apparent in the lens of the proposed framework.

\subsection{Transport by a molecular motor}

Consider a single molecular motor (e.g. kinesin \cite{Howard2001}) attached to a cargo as in \textbf{Figure \ref{fig:motorfig}}(a). Molecular motors produce transport by taking discrete steps along a track, exerting a force on a cargo. A notable feature of these steps is that the rate at which they are taken is well-established to be force-dependent \cite{Carter2005}. That is, motors typically ``slow down'' as force is exerted on them. As this discrete stepping process occurs, another source of noise is the diffusion of the cargo, which is also instantaneously changing the force applied to the motor, and consequently, the rate at which it steps. In our proposed simple motor of transport, we provide a  preliminary investigation into this question.

We assume that the linker between the molecular motor and the cargo can be treated as a Hookean spring, where $x$ denotes the distance stretched (from rest), as depicted in \textbf{Figure \ref{fig:motorfig}}(a). Then, the diffusion component of this process is an Ornstein-Uhlenbeck process. We take the motor to take fixed step sizes, $\Delta$ and the rate at which motors step to be Gaussian $\lambda(x) = \sqrt{\frac{\alpha^2 \beta}{2\pi}}\exp\{ -\beta x^2 /2 \}$. Although there is evidence that the direction of the force is significant for motor stepping, this often deals with the proportion of forward to back steps of the motor. Because we  assume forward steps here, we assume, for simplicity, any force decreases the motor stepping rate. 

\begin{figure}[h]
\includegraphics[width=.225\textwidth]{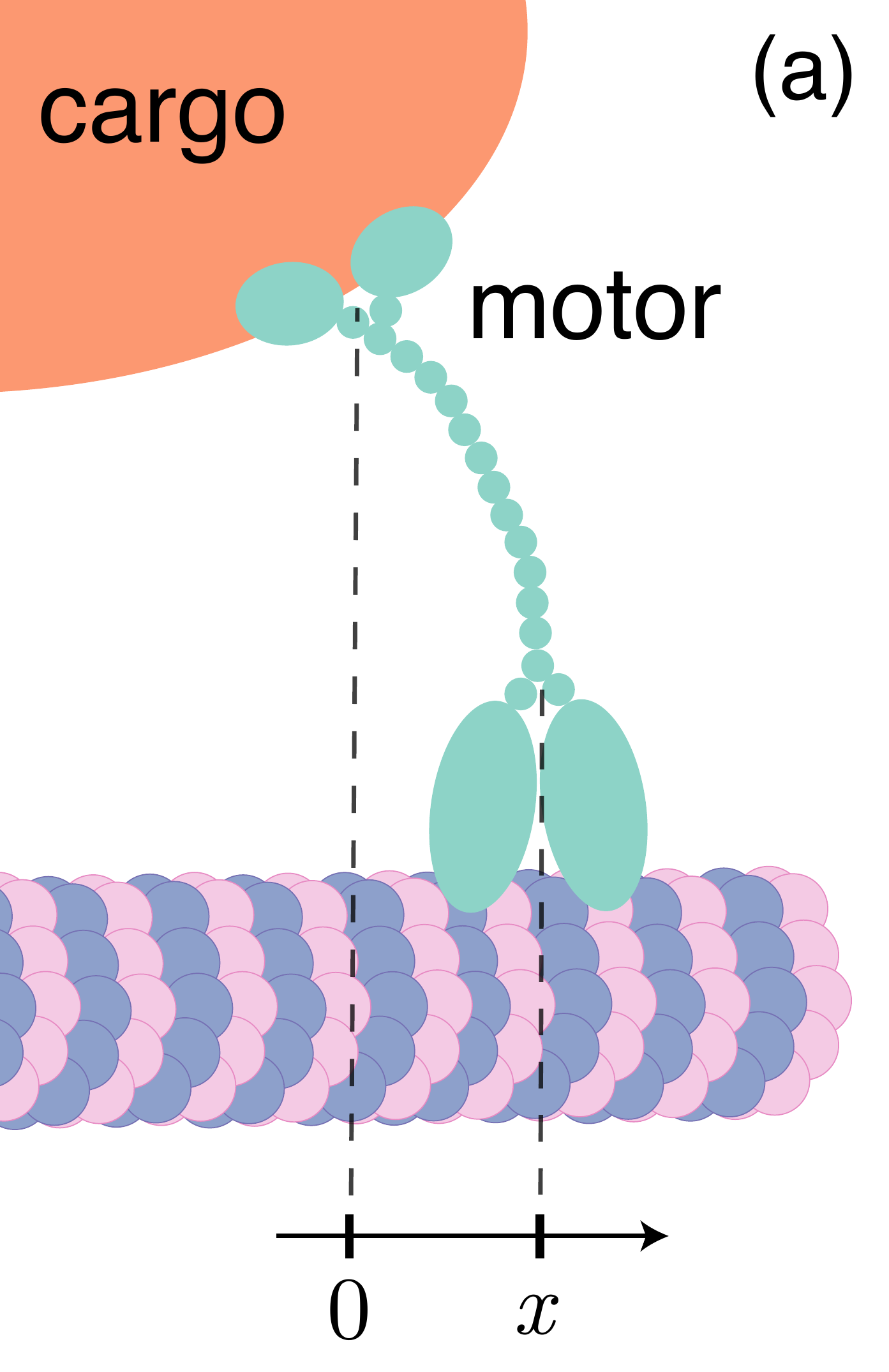}
\includegraphics[width=.775\textwidth]{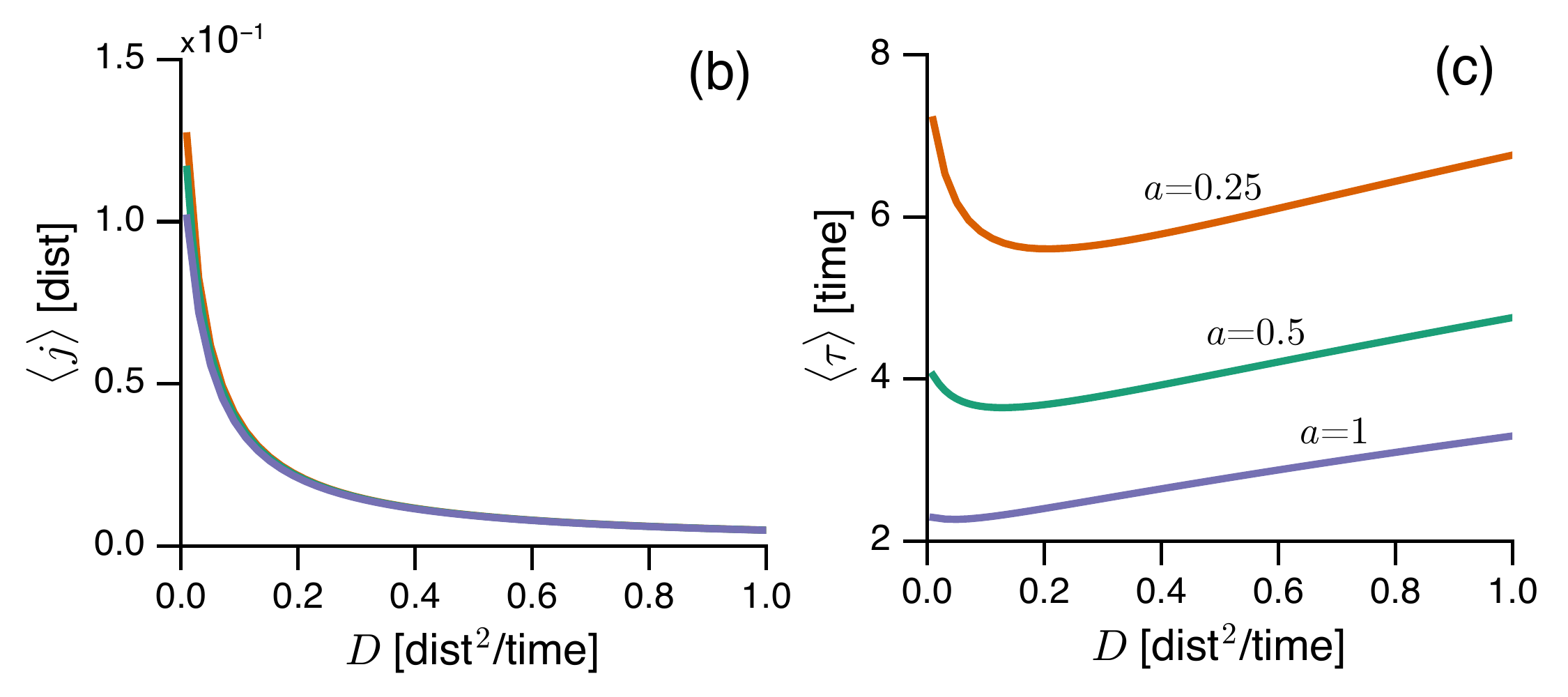}
\protect\caption{\textbf{\sffamily(a)} a diagram illustrating the setup for the molecular motor model, where $x$ denotes the distance the motor linker is stretched from rest. \textbf{\sffamily(b)} the mean stationary jump location distribution for different values of $a$, linker strength as a function of the cargo diffusion coefficient $D$. \textbf{\sffamily(c)} the mean stationary interjump time (stepping rate) for the motor. \label{fig:motorfig}}
\end{figure}

The resulting stationary relationship \eqref{eq:qstar} becomes
\begin{equation}
0 = \partial_x \left\{ axu_\star \right\} + D \partial_{xx} u_\star - \lambda(x) u_\star + \lambda(x-\Delta)u_{\star}(x-\Delta). \label{eq:ustar_ex3}
\end{equation}
There does not appear to be a fruitful approach to studying \eqref{eq:ustar_ex3} analytically. However, we solve can it numerically using a Crank-Nicolson upwind scheme. From this, $u_\star$ and consequently $p_\star$ is obtained. The first moment of the interjump time $\langle \tau_\star \rangle$ is immediate after a quadrature \eqref{eq:taustar}. 

We plot the mean value of $x$ when a jump occurs, denoted $\langle j \rangle$ and the mean time between jumps, $\langle \tau \rangle$ as a function of the cargo diffusion coefficient $D$, for varying values of the linker strength, $a$. These results can be seen in \textbf{Figure \ref{fig:motorfig}}(b,c) for parameter values $D=.1,\alpha=1,\beta=5$. The strength of diffusion appears to produce a curious effect on the other noise source. While the mean jump location, $\langle j \rangle$, remains effectively unchanged for varying values of linker strength, the mean interjump time exhibits very different behavior for different values of $a$. Specifically, we see that if the linker is weak, the interjump time can have a non-monotonic dependence on the diffusion coefficient. That is, the interaction between these two sources of noise (which can be studied as a result of this work) is non-trivial. 

While this model is far simpler than the complex physiological it describes, the overall result (non-monotonic dependence on cargo diffusion) may be a feature that persists in a more detailed motor model (e.g. multi-motor), as other non-monotonic dependences on cargo noise have been seen \cite{Miles2016a}. Although the quantity $\langle j \rangle$ (the average distance a motor is stretched when stepping) is perhaps not the most useful in extracting information about more complex transport, it seems feasible that it it could be utilized in a reward-renewal framework akin to others utilized to study motor transport such as \cite{Krishnan2011,Hughes2011,Hughes2012}. \edit{Motors also often have an asymmetric (state-dependent) unbinding rate, thought to be significant \cite{Takshak2016} and could perhaps provide a future application for this framework.}

\subsection{Independence of $\lambda(x)$}

Lastly, we provide a short example (with no physical motivation in mind) that shows that our framework can  be used with other tools (here, Fourier analysis) to glean interesting features of certain families of models. For this example, we consider a constant drift and diffusion with random jump size distribution $\mu(\Delta)$. Then, in stationarity, \eqref{eq:qstar} becomes
\begin{equation}
	0 = -\partial_x\left\{a u_\star\right\} + D \partial_{xx}u_\star - \lambda(x) u_\star + \int_{-\infty}^{\infty}\mu(\Delta) \lambda(x-\Delta) u_\star(x-\Delta) \, \dd \Delta. \label{eq:uJUMP}
\end{equation}
One quantity of interest is the first moment of jump locations, so 
\begin{equation}
	\langle j_\star \rangle = \int_{-\infty}^{\infty} x p_\star \, \dd x = \int_{-\infty}^{\infty} x u_\star \lambda \, \dd x.
\end{equation}
Although $\lambda(x)$ is arbitrary, surprisingly the Fourier transform can be used to gain some insight. Define the Fourier transform of $u_\star$ to be
\begin{equation}
	U(k) \coloneqq \mathcal{F}[u_\star] = \int_{-\infty}^{\infty} e^{-ikx} u_\star(x) \, \dd x.
\end{equation}
Also define the transformed quantities
\begin{equation}
	\Lambda(k) \coloneqq \mathcal{F}\left[\lambda u_\star \right], \qquad M(k) \coloneqq \mathcal{F}[\mu]
\end{equation}
from which, we note that 
\begin{equation}
 	\Lambda'(0) = -i\langle j \rangle, \qquad M'(0) = - i\langle \mu  \rangle.
 \end{equation} 
 Taking the transform of \eqref{eq:uJUMP},
\begin{equation}
	0 = -aikU - Dk^2U - \Lambda(k) + M(k)\Lambda(k), \label{eq:Uk}
\end{equation}
which, evaluated at $k=0$ yields
\begin{equation}
	\left[1-M(0)\right] \Lambda(0) = 0. \label{eq:1stfourier}
\end{equation}
However, we know  $\mu(\Delta)$ and $u_\star\lambda = p_\star$ are both probability densities and consequently $M(0)=1$ and $\Lambda(0)=1$. Thus, \eqref{eq:1stfourier} is trivially satisfied. This is not useful on its own, but we can couple the higher order moments by taking a $k$ derivative of \eqref{eq:Uk} to yield
\begin{equation}
	0 = -aikU' - aiU - 2DkU - 2k^2U' - \Lambda' + M'\Lambda + M\Lambda', \label{eq:1deriv}
\end{equation}
which, at $k=0$ and using $M(0)=1, \Lambda(0)=1$ yields
\begin{equation}
	0 = -aiU  -\Lambda'(0) + M'(0) + \Lambda'(0),
\end{equation}
which says that necessarily
\begin{equation}
	M'(0) = - i \langle \mu \rangle = aiU(0).
\end{equation}
However, recall from \eqref{eq:tauimean} that $U(0) = \langle \tau_\star \rangle$, the mean interjump time, so 
\begin{equation}
	\langle \tau_\star \rangle = - \frac{\langle \mu \rangle}{a}. \label{eq:taustar_ex2}
\end{equation}
This result is interesting for two reasons. For one, we have the mean interjump time is completely independent of the choice of $\lambda(x)$, aside from the requirement that the system reaches stationarity. Secondly, the sign of \eqref{eq:taustar_ex2}  narrows down the requirements for the process to reach stationarity. Since $\langle \tau_\star \rangle $ is an interjump time, it must be non-negative, meaning this quantity only exists if $a, \langle \mu \rangle$ differ in sign. This is intuitive as the drift and jump process must oppose each other to have a chance of reaching stationarity.  \edit{We acknowledge that the linearity of this example (which allows us to perform the explicit Fourier analysis) also makes the result possible to derive through other methods (e.g. taking expectations of the SDE). However, we hope this analysis conveys that our framework provides a different lens (in junction with elementary tools) to study and possibly reveal behavior of a system.  }

\section{Discussion \& Conclusion}

In this work, we present a general framework for studying jump-diffusion systems with state-dependent jump rates, a class of models that were previously difficult or impossible to study otherwise due to the interaction between the two sources of noise. The formulation is flexible enough to accommodate a variety of behaviors, providing relevance to a wide range of applications. The particular objects of study in this work are the distributions of the jump locations: the values of $X_t$ evaluated at the jump times $t_i$.

We reformulate the full process into a survival description between jumps. Using this reformulation, the explicit distribution of jump locations can be extracted in a manner that depends on the previous jump location. This relation can be articulated as an iterated map, producing a sequence of jump locations. With explicit knowledge of the distribution of jump locations, statistics of the interjump times can be computed in more generality than previous work \cite{Daly2007a}. Taking this map to its limit, or equivalently, assuming the process reaches stationarity, the stationary distribution of the full process and the  jump locations are shown to be closely related. An immediate consequence of this relation is that these two quantities differ if and only if the jump intensity is state dependent. That is, for the class of models studied in this work, the jump locations are distinct from the stationary density. Consequently, if the jump locations are of interest, the results of this work illustrate their relation to observations of the full process. 

Although not the focus of this work, we provide a brief discussion of possible conditions on convergence to stationarity in the novel lens of this framework. This discussion utilizes tools from theoretical ecology motivated by the observation that the iterated map formulation is effectively a non-negative integral linear operator. This is in contrast to previous, more probabilistic approaches to studying finite time blowup of these models or similar \cite{Bally2017, Xi2009,DeVille2014,Cloez2015}. 

Finally, we provide three example applications of the framework. The first application, to a simple model of  stochastic neuronal integrate-and-fire, is used to show a scenario where the jump locations themselves may be of interest. Our minimal model contains no inherent threshold voltage, but, as a product of stochasticity of the firing threshold, forms a robust, sharp peak of firing locations. This provides additional and novel support for the use of a deterministic threshold as in classical integrate-and-fire and its extensions. The second model, one of intracellular transport by a molecular motor is used to demonstrate how the two sources of noise (diffusion and jumps) interact. We ultimately find that the mean stepping rate of the motor can have a non-monotonic dependence on the strength of cargo diffusion, a consequence of the interaction between the two noise sources. Finally, a third example calculation shows that certain classes of models (in this case, constant drift, random jump size) can have curious behaviors. Specifically, we find that the mean interjump time for be \textit{independent} of the choice of $\lambda(x)$, so long as the process does indeed reach stationarity. 

This work proposes a preliminary framework for studying a wide class of problems. Although not explicitly studied here, we suspect that with more exotic diffusions (e.g. colored noise, which fits within the scope of these results), the interactions between the feedback between the noise sources can produce even more interesting behavior and will be explored in the future. Additionally, a more precise exploration is necessary of how an iterated map framework such as this can be used to provide conditions of stationarity.

\section{Acknowledgments}
This research was partially supported by NSF grant DMS 1515130 and DMS-RTG 1148230.
\appendix
\section{Higher order interjump time moments} \label{sect:higher_moments}
In theory, higher order moments of the interjump times can be computed with knowledge of the sequence of jump distributions $p_i$ as was discussed with the first moment in \eqref{eq:tauimean}. 

To illustrate this, consider \eqref{eq:ptaui} and take the second moment with respect to $\tau_i$ on both sides, resulting in
\begin{equation}
\langle \tau_i^2 \rangle = -\int_{-\infty}^{\infty} \int_0^{\infty} t^2 \partial_t \hat{p}_i \, \dd t \, \dd x.
\end{equation}
Again, noting that $u_i = p_i/\lambda = \int_0^{\infty} \hat{p}_i \, \dd t$, after integrating by parts, we get
\begin{equation}
\langle \tau_i^2 \rangle = -\int_{-\infty}^{\infty} v_i \, \dd x,
\end{equation}
which looks similar to \eqref{eq:tauimean}, however, $v_i$ is defined by the relation 
\begin{equation}
[\lambda(x) - \Ell]v_i =  u_i. \label{eq:v_i}
\end{equation}
Thus, computation of the second moment requires solving the differential relationship \eqref{eq:v_i}, rather than just a quadrature as in the first moment. Continuing to higher order moments using same approach yields a further hierarchy coupled in a differential manner. Interestingly, the differential operator (left hand side) of \eqref{eq:v_i} is exactly $\TT$, the same as \eqref{eq:qi}, meaning the Green's function (effectively $\TT^{-1}$, discussed more in \textbf{\ref{sect:spectral}}) could  be used to compute these quantities. 

\section{Spectral properties of iterated map} \label{sect:spectral}
Note that \eqref{eq:qi} involves the inverse of a differential linear operator $\TT$, which is exactly determined by its corresponding Green's function. Let $G(x,\xi)$ be the corresponding Green's function to $\TT$, meaning that
\begin{equation}
	\TT_x G(x,\xi) = \delta(x-\xi).
\end{equation}
Then, \eqref{eq:qi} can be rewritten as
\begin{equation}
	u_{i+1} = \int_{-\infty}^{\infty} G(x,\xi) \JJ \lambda u_i(\xi) \, \dd \xi.	
\end{equation}
After a change of variables, define $\tilde{G}$ by
\begin{equation}
u_{i+1} = \int_{-\infty}^{\infty} \tilde{G}(x,\zeta) u_i(\zeta) \, \dd \zeta \coloneqq \int_{-\infty}^{\infty} G(x,\xi) \JJ \lambda u_i(\xi) \, \dd \xi. \label{eq:tildeG}
\end{equation}

For example, for fixed jump sizes: $\JJ p = p(x-\Delta)$, then 
\begin{equation}
	\tilde{G}(x,\zeta) = G(x, \zeta+\Delta)\lambda(\zeta+\Delta).
\end{equation}
Abbreviate this linear non-negative integral operator
\begin{equation}
\mathbb{A} q \coloneqq \int_{-\infty}^{\infty} \tilde{G}(x,\zeta) q(\zeta) \, \dd \zeta.
\end{equation}

Now, iterations of our map correspond to iterating the integral operator $\mathbb{A}$ with kernel $\tilde{G}$. This formulation is exactly that of the so-called \textit{integral projection models} (IPM) in theoretical ecology\cite{Ellner2006}. Stability results from the IPM literature can be used to understand the convergence of our iterative procedure.

The main result comes from \cite{Easterling1998} and is effectively a statement of the Krein-Rutman theorem \cite{Krein1950,Krasnoselskii1964}, the infinite dimensional analog of the Perron-Frobenius theorem.  \cite{Keener2000,Bressloff2014}

Although $L^1$ seems like the natural function space to study the spectral properties of these operators, since they must preserve probability, $L^2$ turns out to be far more accessible due to issues establishing compactness in $L^1$. In Appendix C of \cite{Ellner2006}, the authors provide a more thorough commentary on these complications. we cite the main theorem which establishes the existence of a dominant eigenvalue for $\Aop$. 
\begin{theoremA}[Easterling 1998, \cite{Easterling1998}]
Suppose that $\tilde{G} \in L^2$ and is non-negative. If  there exists an $\alpha >0 , \beta >0, u_0$ such that
\begin{equation}
\alpha(x) u_0 (\xi) \leq \tilde{G}(x,\xi) \leq \beta(x) u_0 (\xi)
\end{equation} 
for all $x,\xi$, then the integral operator $\Aop$ has a dominant eigenvalue with associated eigenfunction. 
\end{theoremA}
This establishes conditions for the existence of a dominant eigenvalue and eigenvector, which establish the long-term behavior. 
\begin{theoremA}[Easterling 1998, \cite{Easterling1998}]
Assuming that $\Aop$ satisfies the previous condition, then the stationary distribution $u_\star$ is given by the eigenfunction $\phi_1$ associated with the dominant eigenvalue $\lambda_1$, 
\begin{equation}
\lim_{i\to\infty} \frac{u_i}{\lambda_1^i} = \kappa \phi_1,
\end{equation}
where $\kappa$ is a scaling parameter.
\end{theoremA}
The previous theorem suggests that the sequence $u_i$ converges if and only if $|\lambda_1| = 1$. Thus, this summary provides a rough heuristic, but possibly novel angle for determining whether the sequence $u_i$ (and the full process) approach stationarity, although we reemphasize that this is not our focus.

\bibliography{jumpdiffbib}


\end{document}